\documentclass[letterpaper, 10 pt, conference]{ieeeconf}  

\IEEEoverridecommandlockouts                              
\usepackage{cite}
\usepackage{times}
\usepackage{mathptmx}       
\usepackage{mathtools}
\usepackage{rotating}
\usepackage{amsmath}
\usepackage{amsfonts,amssymb}
\usepackage{fancybox,color}
\usepackage{algorithm}
\usepackage{algorithmic}
\usepackage{url}
\usepackage{subfigure}

\newtheorem{assumption}{Assumption}
\newtheorem{theorem}{Theorem}
\newtheorem{proposition}{Proposition}

\title{\LARGE \bf
Robotic Manifold Tracking of Coherent Structures in Flows
}

\author{M. Ani Hsieh, Eric Forgoston, T. William Mather, and Ira B. Schwartz
\thanks{This work was supported by the Office of Naval Research (ONR). MAH was an ONR Summer Faculty Fellow and supported by ONR Contract No. N0001411WX20079.  EF is supported by the Naval Research Laboratory (Award No.
N0017310-2-C007)  TWM is supported by the National Science Foundation under Grant No. DGE-0947936.}
\thanks{M. Ani Hsieh and T. William Mather are with the SAS Laboratory, Mechanical Engineering \& Mechanics Department,
        Drexel University, Philadelphia, PA 19104, USA
        {\tt\small \{mhsieh1,twm32\}@drexel.edu}}%
\thanks{Eric Forgoston is with the Department of Mathematical Sciences, Montclair State
        University, Montclair, New Jersey 07043, USA
        {\tt\small eric.forgoston@montclair.edu}}%
\thanks{Ira B. Schwartz is with the Nonlinear Systems Dynamics Section, Plasma Physics Division, Code 6792, U.S. Naval Research Laboratory, Washington, DC 20375, USA
        {\tt\small ira.schwartz@nrl.navy.mil}}%
}

\begin{document}

\maketitle
\thispagestyle{empty}
\pagestyle{empty}

\begin{abstract}
Tracking Lagrangian coherent structures in dynamical systems is important for many applications such as oceanography and weather prediction. In this paper, we present a collaborative robotic control strategy designed to track stable and unstable manifolds. The technique does not require global information about the fluid dynamics, and is  based on local sensing, prediction, and correction.  The collaborative control strategy is implemented on a team of three robots to track coherent structures and manifolds on static flows as well as a noisy time-dependent model of a wind-driven double-gyre often seen in the ocean.  We present simulation and experimental results and discuss theoretical guarantees of the collaborative tracking strategy.
\end{abstract}

\section{INTRODUCTION}\label{sec:intro}

In this paper, we present a collaborative control strategy for a class of autonomous underwater vehicles (AUVs) to track the coherent structures and manifolds on flows.  In realistic ocean flows, these time-dependent coherent structures, or Lagrangian coherent structures (LCS), are similar to separatrices that divide the flow into dynamically distinct regions.  LCS are extensions of stable and unstable manifolds to general time-dependent flows \cite{ref:Haller2000}, and they carry a great deal of global information about the dynamics of the flows. For two-dimensional (2D) flows, LCS are analogous to ridges defined by local maximum instability, and quantified by local measures of Finite-Time Lyapunov Exponents (FTLE) \cite{ref:Shadden2005}.

Recently, LCS have been shown to coincide with optimal trajectories in the ocean which minimize the energy and the time needed to traverse from one point to another \cite{ref:Inanc2005, ref:Senatore2008}.  Furthermore, to improve weather and climate forecasting, and to better understand various physical, chemical, and geophysical processes in the ocean, there has been significant interest in the deployment of autonomous sensors to measure a variety of quantities of interest.  One drawback to operating sensors in time-dependent and stochastic environments like the ocean is that the sensors will tend to escape from their monitoring region of interest.  Since the LCS are inherently unstable and denote regions of the flow where more escape events may occur \cite{fbys11}, knowledge of the LCS are of paramount importance in maintaining a sensor in a particular monitoring region.

Existing work in cooperative boundary tracking for robotic teams that relies on one-dimensional (1D) parameterizations include~\cite{Hsieh2005, ref:Susca2008} and~\cite{Triandaf2005, ref:Goncalves2010} for static and time-dependent cases respectively. Formation control strategies for
distributed estimation of level surfaces and scalar fields in the ocean are presented in~\cite{ref:Zhang2007, lynch2008, ref:Wu2011} and pattern formation for surveillance and monitoring by robot teams is discussed in~\cite{ref:Spletzer2005, ref:Kalantar06, ref:Hsieh2007}.  Our work is distinguished from existing work in that we use cooperative robots to find coherent structures without requiring a global picture of the ocean dynamics.  We take inspiration from \cite{ref:Nusse1989} and design a strategy to enable a team of robots to track the stable/unstable manifolds of general 2D conservative flows through local sensing alone.  We verify the feasibility of our method through simulations and experiments and show how the proposed strategy can be extended to track coherent structures in time-dependent conservative flows with measurement noise.  To our knowledge, this is the first attempt in the development of tracking strategies for mapping LCS in the ocean using AUVs.

The novelty of this work lies in the use of nonlinear dynamical and chaotic system analysis techniques to derive a tracking strategy for a team of robots. The cooperative control strategy leverages the spatio-temporal sensing capabilities of a team of networked robots to track the boundaries separating the regions in phase space that support distinct dynamical behavior.  Additionally, our boundary tracking relies solely on local measurements of the velocity field.  Our technique is quite general, and may be applied to any conservative flow.

The paper is structured as follows: We formulate the problem and outline key assumptions in Section \ref{sec:probForm}.  The cooperative control strategy is presented in Section \ref{sec:method} and its theoretical properties analyzed in Section \ref{sec:analysis}.  Section \ref{sec:results} presents our simulation and experimental results.  The extension of the proposed strategy to a noisy time-dependent model of a wind-driven double-gyre is presented in \ref{sec:doubleGyre}.  We conclude with a discussion of our results and directions for future work in Sections \ref{sec:discuss} and \ref{sec:future} respectively.

\section{PROBLEM FORMULATION}\label{sec:probForm}
We consider the problem of controlling a team of $N$ planar AUVs to collaboratively track the material lines that separate regions of flow with distinct fluid dynamics.  This is similar to the problem of tracking the stable (and unstable) manifolds of a general nonlinear dynamical system where the manifolds separate regions in phase space with distinct dynamical behaviors.  We assume the following 2D kinematic model for each of the AUVs:
\begin{subequations} \label{eq:kinematics}
\begin{equation}
\dot{x}_i = V_i \cos\theta_i + u_{i},
\end{equation}
\begin{equation}
\dot{y}_i = V_i \sin\theta_i + v_{i},
\end{equation}
\end{subequations}
where $\mathbf{x}_i = [x_i, \, y_i]^T$ is the vehicle's planar position, $V_i$ and $\theta_i$ are the vehicle's linear speed and heading, and $\mathbf{u}_i = [u_i, \, v_i]^T$ is the velocity of the fluid current experienced/measured by the $i^{th}$ vehicle.  Additionally, we assume each agent can be circumscribed by a circle of radius $r$, {\it i.e.}, each vehicle can be equivalently described as a disk of radius $r$.

In this work, $\mathbf{u}_i$ is provided by a 2D planar conservative vector field described by a differential equation of the form
\begin{equation}\label{eq:flowMap}
\mathbf{\dot{x}} = F(\mathbf{x}).
\end{equation}
In essence, $u_i = F_x(\mathbf{x}_i)$ and $v_i = F_y(\mathbf{x}_i)$.  Let $B_S$ and $B_U$ denote the stable and unstable manifolds of \eqref{eq:flowMap}.  In general, $B_S$ and $B_U$ are the separating boundaries between regions in phase space with distinct dynamics.  For 2D flows, $B_*$ are simply one-dimensional curves where $*$ denotes either stable ($S$) or unstable ($U$) boundaries.  For a small region centered about a point on $B_*$, the system is unstable in one dimension.  Finally, let $\rho(B_*)$ denote the radius of curvature of $B_*$ and assume that the minimum of the   radius of curvature $\rho_{min}(B_*) > r$.  This last assumption is needed to ensure the robots do not lose track of the $B_*$ due to sharp turns.

The objective is to develop a collaborative strategy to enable a team of robots to track $B_*$ in general 2D planar conservative flow fields through local sampling of the velocity field.  While the focus is on the development of a tracking strategy for $B_S$, the proposed method can be easily extended to track $B_U$ since $B_U$ are simply stable manifolds of \eqref{eq:flowMap} for $t < 0$. We present our methodology in the following section.

\section{METHODOLOGY}\label{sec:method}
Our methodology is inspired by the Proper Interior Maximum (PIM) Triple Procedure \cite{ref:Nusse1989} -- a numerical technique designed to find stationary trajectories in chaotic regions with no attractors.  While the original procedure was developed for chaotic dynamical systems, the approach can be employed to reveal the stable set of a saddle point of a general nonlinear dynamical system.  The procedure consists of iteratively finding an appropriate PIM Triple on a saddle straddling line segment and propagating the triple forward in time.  We briefly summarize the procedure in the following section and refer the interested reader to \cite{ref:Nusse1989} for further details.

\subsection{The PIM Triple Procedure}
Given the dynamical system described by \eqref{eq:flowMap}, let ${\cal D} \in \mathbb{R}^2$ be a closed and bounded set such that ${\cal D}$ does not contain any attractors of \eqref{eq:flowMap}.  Given a point $\mathbf{x} \in {\cal D}$, the escape time of $\mathbf{x}$, denoted by $T_E(\mathbf{x})$, is the time $\mathbf{x}$ takes to leave the region ${\cal D}$ under the differential map \eqref{eq:flowMap}.

Let $J$ be a line segment that crosses the stable set $B_S$ in ${\cal D}$, {\it i.e.}, the endpoints of the $J$ are on opposite sides of $B_S$.  Let $\{\mathbf{x}_L, \mathbf{x}_C, \mathbf{x}_R\}$ denote a set of three points in $J$ such that $\mathbf{x}_C$ denotes the interior point.  Then $\{\mathbf{x}_L, \mathbf{x}_C, \mathbf{x}_R\}$ is an {\it Interior Maximum} triple if $T_E(\mathbf{x}_C) > \max\{T_E(\mathbf{x}_L), T_E(\mathbf{x}_R)\}$. Furthermore, $\{\mathbf{x}_L, \mathbf{x}_C, \mathbf{x}_R\}$ is a {\it Proper Interior Maximum} (PIM) triple if it is an Interior Maximum triple and the interval $[\mathbf{x}_L, \mathbf{x}_R]$ in $J$ is a proper subset of $J$.

Then the numerical computation of any PIM triple can be obtained iteratively
starting with an initial saddle straddle line segment $J_0$. Let
$\mathbf{x}_{L_0}$ and $\mathbf{x}_{R_0}$ denote the endpoints of $J_0$ and
apply an $\epsilon_0 > 0$ discretization of $J_0$ such that $\mathbf{x}_{L_0}
= \mathbf{q}_0 < \mathbf{q}_1 < \ldots < \mathbf{q}_M = \mathbf{x}_{R_0}$.
For every point $\mathbf{q}_i$, determine $T_E(\mathbf{q}_i)$ by propagating
$\mathbf{q}_i$ forward in time using \eqref{eq:flowMap}. Then the PIM triple
in $J_0$ is given by the the points $\{\mathbf{q}_{k-1}, \mathbf{q}_k,
\mathbf{q}_{k+1}\}$ where $\mathbf{q}_k = \arg\max\limits_{i=1, \ldots, M}
T_E(\mathbf{q}_i)$. This PIM triple can then be further refined by choosing $J_1$ to be the line segment containing $\{\mathbf{q}_{k-1}, \mathbf{q}_k, \mathbf{q}_{k+1}\}$ and reapplying the procedure with another $\epsilon_1 > 0$ discretization where $\epsilon_1 < \epsilon_0$.

Given an initial saddle straddling line segment $J_0$, it has been shown that
the line segment given by any subsequent PIM triple on $J_0$ is also a saddle
straddling line segment \cite{ref:Nusse1989}.  Furthermore, if we use a PIM
triple $\mathbf{x}(t) = \{\mathbf{x}_L, \mathbf{x}_C, \mathbf{x}_R\}$ as the
initial conditions for the dynamical system given by \eqref{eq:flowMap} and
propagate the system forward in time by $\Delta t$, then the line segment
containing the set $\mathbf{x}(t+\Delta t)$, $J_{t+\Delta t}$, remains a
saddle straddle line segment. As such, the same numerical procedure can be
employed to determine an appropriate PIM trip on $J_{t+\Delta t}$.  This procedure can be repeated to eventually reveal the entire stable set $B_{S}$ and unstable set $B_{U}$ within ${\cal D}$ if time was propagated forwards and backwards respectively.  Furthermore, since the procedure always begins with a valid saddle straddling line segment, by construction, the procedure always results in a non-empty set.

Inspired by the PIM Triple Procedure, we propose a cooperative {\it saddle straddle} control strategy for a team of $N\geq3$ robots to track the stable (and unstable) manifolds of a general conservative time-independent flow field $F(\mathbf{x})$.  Different from the procedure, our robots will solely rely on information that can be gathered via local sensing and shared through the network.  In contrast, a straight implementation of the PIM Triple Procedure would require global knowledge of the structure of the system dynamics throughout a given region given its reliance on computing escape times. We describe our approach in the following section.

\subsection{Controller Synthesis}
Consider a team of three robots and identify them as robots $\{L, C, R\}$.  While the robots may be equipped with similar sensing and actuation capabilities, we propose a heterogeneous cooperative control strategy.

Let $\mathbf{x}(0) = [\mathbf{x}_{L}^{T}(0), \, \mathbf{x}_{C}^{T}(0), \,
\mathbf{x}_{R}^{T}(0)]^{T}$ be the initial conditions for the three robots.
Assume that $\mathbf{x}(0)$ lies on the line segment $J_0$ where $J_0$ is a
saddle straddle line segment and $\{\mathbf{x}_L(0), \mathbf{x}_C(0),
\mathbf{x}_R(0)\}$ constitutes a PIM triple.  Similar to the PIM Triple
Procedure, the objective is to enable the robots to  maintain a formation such
that a valid saddle straddle line segment can be maintained between robots $L$
and $R$.  Instead of computing the escape
times for points on $J_0$ as proposed by the PIM Triple Procedure, robot $C$ must remain
close to $B_S$ using only local measurements of the velocity field 
provided by the rest of the team.  As such, we refer to robot $C$ as the
tracker of the team while robots $L$ and $R$ maintains a {\it straddle
  formation} across the boundary at all times.  Robots $L$ and $R$ may be thought of herding robots, since they control and determine the actions of the tracking robot.

\subsubsection{Straddling Formation Control}
The controller for the straddling robots consists of two discrete states: a
passive control state, $U_{P}$, and an active control state, $U_{A}$.  The
robots initialize in the passive state $U_{P}$ where the objective is to
follow the flow of the ambient vector field. Therefore, $V_i = 0$ for $i = L, R$.  Robots execute $U_{P}$ until they reach the maximum allowable separation distance $d_{Max}$ from robot $C$.  When $\Vert \mathbf{x}_i - \mathbf{x}_C\Vert > d_{Max}$, robot $i$ switches to the active control state, $U_{A}$, where the objective is to navigate to a point $\mathbf{p}_i$ on the current {\it projected} saddle straddle line segment $\hat{J}_t$ such that $\Vert \mathbf{p}_i - \mathbf{p}_c\Vert = d_{Min}$ and  $\mathbf{p}_C$ denotes the midpoint of $\hat{J}_t$.  When robots execute $U_{A}$ , $V_i = \Vert (\mathbf{p}_i - \mathbf{x}_i) - \mathbf{u}_i\Vert$ and $\theta_i(t) = \alpha_i(t)$ where $\alpha_i$ is the angle between the desired, $(\mathbf{p}_i - \mathbf{x}_i)$, and current heading, $\mathbf{u}_i$, of robot $i$ as shown in Fig. \ref{fig:schematic}.  In summary, the straddling control strategy for robots $L$ and $R$ is given by
\begin{subequations} \label{eq:straddleControl}
\begin{align}
V_i & = \left\{
    \begin{array}{ll}
        0 & \textrm{if } d_{Min} < \Vert \mathbf{x}_i - \mathbf{x}_C\Vert < d_{Max} \\
    \Vert (\mathbf{p}_i - \mathbf{x}_i) - \mathbf{u}_i\Vert & \textrm{otherwise}
    \end{array}\right.,\\
\theta_i & = \left\{
    \begin{array}{ll}
        0 & \textrm{if } d_{Min} < \Vert \mathbf{x}_i - \mathbf{x}_C\Vert < d_{Max} \\
        \alpha_i & \textrm{otherwise}
    \end{array}\right.~.
\end{align}
\end{subequations}

\begin{figure}
\centering
\includegraphics[width=0.7\linewidth]{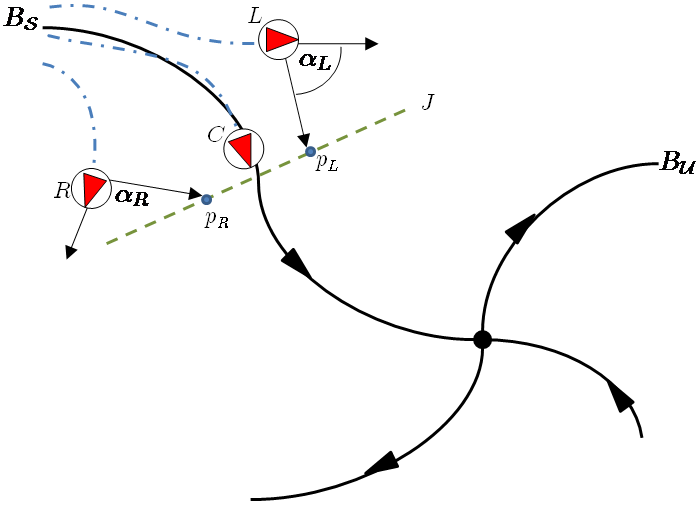}
\caption{Three robots tracking $B_{S}$ in a conservative vector field.  The blue dash-dot lines represent the robot trajectories, the green dashed line represents the saddle straddle line segment $J$, and $\mathbf{p}_L$ and $\mathbf{p}_R$ denotes the target positions for $L$ and $R$ respectively when executing $U_{P}$ and $U_{A}$.\label{fig:schematic}}
\end{figure}

We note that while the primary control objective for robots $L$ and $R$ is to maintain a straddle formation across $B_S$, robots $L$ and $R$ are also constantly sampling the velocity of the local vector field and communicating these measurements and their relative positions to robot $C$.  Robot $C$ is then tasked to use these measurements to track the position of $B_{S}$.


\subsubsection{Manifold Tracking Control}
Let $\mathbf{\hat{u}}_{L}(t)$, $\mathbf{\hat{u}}_{C}(t)$, and
$\mathbf{\hat{u}}_{C}(t)$ denote the current velocity measurements obtained by
robots $L$, $C$, and $R$ at their respective positions. Let $d(\cdot, \cdot)$
denote the Euclidean distance function and assume that $d(\mathbf{x}_C, B_S) <
\epsilon$ such that $\epsilon > 0$ is small.  Given the straddle line
segment $J_t$ such that $\mathbf{x}_L(k)$ and $\mathbf{x}_R(k)$ are the
endpoints of $J_t$, we consider an $\epsilon_t < \epsilon$ discretization of
$J_t$ such that  $\mathbf{x}_L = \mathbf{q}_1 < \mathbf{q}_2 < \ldots <
\mathbf{q}_M = \mathbf{x}_R$.  The objective is to use the velocity
measurements provided by the team to interpolate the vector field at the
points $\mathbf{q}_1, \ldots, \mathbf{q}_M$.  Since \eqref{eq:flowMap} has
${\cal C}^1$ continuity and if $\mathbf{x}_C$ is $\epsilon$-close to $B_S$,
then the point $\mathbf{q}_B = \arg\max\limits_{k=1, \ldots, M}
\mathbf{u}(\mathbf{q}_k)^T \mathbf{\hat{u}}_C(t)$ should be $\delta$-close to
$B_S$ where $\epsilon < \delta < A$ and $A$ is a small enough positive constant.

While there are numerous vector field interpolation techniques available \cite{ref:Agui1987, ref:Marchioli2007, ref:Fuselier2009}, we employ the {\it inverse distance weighting method} described in \cite{ref:Agui1987}.  For a given set of velocity measurements $\mathbf{\hat{u}}_i(t)$ and corresponding position estimates $\mathbf{\hat{x}}_i(t)$, the velocity vector at some point $\mathbf{q}_k$ is given by
\begin{eqnarray*}
\mathbf{u}(\mathbf{q}_k) = \sum_{j} \sum_{i=1}^{N} \frac{w_{ij}\mathbf{\hat{u}}_i(j)}{\sum_{j} \sum_{i=1}^{N} w_{ij}}
\end{eqnarray*}
where $w_{ij} = \Vert \mathbf{\hat{x}}_i(j) - \mathbf{q}_i\Vert^{-2}$.  Rather
than rely solely on the current measurements provided by the three robots, it
is possible to include the recent history of $\mathbf{\hat{u}}_i(t)$ to
improve the estimate of $\mathbf{u}(\mathbf{q}_k)$, {\it i.e.},
$\mathbf{\hat{u}}_i(t-\Delta T)$, $\mathbf{\hat{u}}_i(t-2\Delta T)$, and so
on, where $\Delta T$ is the sampling period and $i=\{L,C,R\}$.  Thus, the control strategy for the tracking robot $C$ is given by
\begin{subequations}\label{eq:trackControl}
\begin{align}
V_C & = \Vert [(\mathbf{q}_B + b \mathbf{\hat{u}}_B) - \mathbf{x}_C] - \mathbf{u}_C\Vert\\
\theta_C & = \beta_C
\end{align}
\end{subequations}
where $\beta_C$ denotes the difference in the heading of robot $C$ and the
vector $(\mathbf{q}_B - \mathbf{\hat{u}}_B)$ and $b>r$ is a small number.  The
term $b \mathbf{\hat{u}}_B$ is included to ensure that the control strategy
aims for a point in front of robot $C$ rather than behind it.  As such, the {\it projected} saddle straddle line segment $\hat{J}_t$
at each time step is given by $\mathbf{p}_c = q_C + b\mathbf{u}_C$ with 
$\hat{J}_t$ orthogonal to $B_S$ at $q_C$ and $\Vert \hat{J}_t \Vert$ chosen to be in the interval $[2d_{Min}, 2d_{Max}]$.

\section{ANALYSIS}\label{sec:analysis}
In this section, we discuss the theoretical feasibility of the proposed saddle straddle control strategy.  We begin with the following key assumption on the robots' initial positions.

\begin{assumption} Given a team of three robots $\{L, C, R\}$, assume that $d(\mathbf{x}_C(0), B_S) < \epsilon$ for a small value of $\epsilon > 0$, $\Vert \mathbf{x}_L - \mathbf{x}_C \Vert = \Vert \mathbf{x}_R - \mathbf{x}\Vert = d_{Min}$ with $d_{Min} > 2r$, and robots $L$ and $R$ are on opposite sides of $B_S$.
\end{assumption}

In other words, assume that the robots initialize in a valid PIM triple formation and their positions form a saddle straddle line segment orthogonal to $B_S$.  Our main result concerns the validity of the saddle straddle control strategy.

\begin{theorem}\label{thm:1} Given a team of $3$ robots with kinematics given by \eqref{eq:kinematics} and $\mathbf{u}_i$ given by \eqref{eq:flowMap}, the feedback control strategy \eqref{eq:straddleControl} and \eqref{eq:trackControl} maintains a valid saddle straddle line segment in the time interval $[t, t+\Delta t]$ if the initial positions of the robots, $\mathbf{x}(t)$, is a valid PIM triple.
\end{theorem}

\begin{proof}
To show this, we must show that at time  $t+\Delta t$, robots $L$ and $R$ remain on opposite sides of $B_S$.  Consider the rate of change of the following function
\begin{align*}
H(\mathbf{x}_L, \mathbf{x}_R) & = \frac{1}{2}(\mathbf{x}_L-\mathbf{x}_R)^{T}(\mathbf{x}_L-\mathbf{x}_R).
\end{align*}
The above expression is simply one half the square of the distance between robots $L$ and $R$.  Let $J_t$ denote the saddle straddle line segment defined by $\mathbf{x}_L(t)$ and $\mathbf{x}_R(t)$ at $t$ and let $\mathbf{p}_B$ be the intersection of $J_t$ and $B_S$.  By construction, if we linearize \eqref{eq:flowMap} about the point $\mathbf{p}_B$, then the Jacobian of \eqref{eq:flowMap} at $\mathbf{p}_B$ will have one positive eigenvalue.  Furthermore, the linearized system can be diagonalized such that the direction of instability lies along $J_t$ \cite{ref:Hartman2002}. Thus, $\frac{d}{dt}H > 0$ in the time interval $[t, t+\Delta t]$ when $V_i=0$ in \eqref{eq:straddleControl}.

When $V_i \neq 0$ in \eqref{eq:straddleControl} for $i = L, R$, $\frac{d}{dt}H < 0$ if the robots $L$ and $R$ are moving closer to robot $C$ after reaching the maximum allowable separation distance.  Recall $\rho_{min}(B_S) > r$, the smallest radius of curvature of $B_S$, and $d_{Min} > 2r$.  Furthermore, robot $C$ initializes $\epsilon$-close to the boundary and \eqref{eq:trackControl} steers $C$ towards $\mathbf{p}_C$ on $\hat{J}_t$ where $\hat{J}_t$ is orthogonal to $B_S$ at $\mathbf{x}_C$.  This ensures that the rate of the change of the radius of curvature of the manifold $B_S$ is small enough such that $\hat{J}_t$ intersects with $B_S$ only once.  Since $d_{Min} > 2r$, this ensures that even if $\frac{d}{dt}H < 0$, the straddling robots never cross the boundary as they move closer to the tracking robot.
\end{proof}

The above theorem guarantees that for any given time interval $[t, t+ \Delta t]$ the team maintains a valid PIM triple formation.  As such, the iterative application of the proposed control strategy leads to the following proposition.

\begin{proposition}
Given a team of $3$ robots with kinematics given by \eqref{eq:kinematics}
and $\mathbf{u}_i$ given by \eqref{eq:flowMap}, the feedback control strategy
\eqref{eq:straddleControl} and \eqref{eq:trackControl} results in an estimate
of $B_S$, denoted as $\hat{B}_S$, such that $\left<B_S, \hat{B}_S
\right>_{L_2} < W$ for some $W > 0$ where $\left<\cdot, \cdot\right>_{L_2}$ denotes the
inner product (which provides an $L_2$ measure between the $B_S$ and
  $\hat{B}_S$ curves).
\end{proposition}

From Thm. \ref{thm:1}, since the team is able to maintain a valid PIM triple formation across $B_S$ for any given time interval $[t, t+\Delta t]$, this ensures that an estimate of $B_S$ in the given time interval also exists.  Applying this reasoning in a recursive fashion, one can show that an estimate of $B_S$ can be obtained for any arbitrary time interval.  The challenge, however, lies in determining the bound on $W$ such that $\hat{B}_S$ results in a good enough approximation since $W$ depends on the sensor and actuation noise, the vector interpolation routine, the sampling frequency, and the time scales of the flow dynamics.  This is a direction for future work.

\section{RESULTS}\label{sec:results}
\subsection{Simulations}
We illustrate the proposed control strategy given by \eqref{eq:straddleControl} and \eqref{eq:trackControl} with some simulation results. Fig. \ref{fig:sine} shows the trajectories of three robots tracking a sinusoidal boundary while Fig. \ref{fig:star} shows the team tracking a 1D star-shaped boundary.  We note that throughout the entire length of the simulation, the team maintains a saddle straddle formation across the boundary.

In both examples, $\mathbf{u} = -a \nabla \varphi - b \nabla \times \psi$ where $a, b > 0$ and $\varphi$ is an artificial potential function such that $\varphi(\mathbf{x}) = 0$ for all $\mathbf{x} \in B_*$ and $\varphi(\mathbf{x}) < 0$ for any $\mathbf{x} \in \mathbb{R}^2/B_*$.  The vector $\psi$ is a $3 \times 1$ vector whose entries are given by $[0, \, 0, \,\gamma(x,y)]^T$ where $\gamma(x,y)$ is the curve describing the desired boundary \cite{ref:Hsieh2007}.  Lastly, the estimated position of the boundary is given by the position of the tracking robot, {\it i.e.}, robot $C$.  In these examples, we filtered the boundary position using a simple first-order low pass filter.

\begin{figure}
\centering
\subfigure[]{\includegraphics[width=0.45\linewidth]{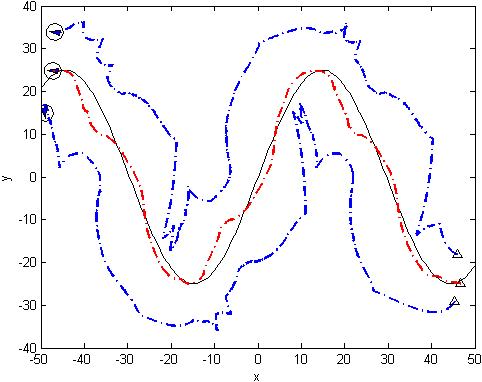}\label{fig:sine}}
\subfigure[]{\includegraphics[width=0.45\linewidth]{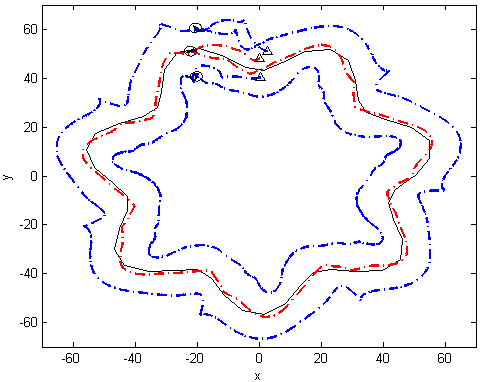}\label{fig:star}}
\caption{Trajectories of $3$ robots tracking a (a) sinusoidal boundary and a (b) star-shaped boundary.  The red dashed line is the estimated position of the desired boundary shown in solid black.  The start positions are shown by $\triangle$ and the end positions are shown by the circle-enclosed blue triangles.}
\end{figure}

\subsection{Experiments}
We also implemented the control strategy on our multi-robot testbed.  The testbed consisted of three mSRV-1 robots in a 4.8x5.4 meter workspace. The mSRV-1 are differential-drive robots equipped with an embedded processor, color camera, and 802.11 wireless capability.  Localization for each robot was provided via a network of overhead cameras.  Fig. \ref{fig:trajectories} shows the trajectories of the robots tracking a star shaped boundary shown in black.  Fig. \ref{fig:expFrame} is a snapshot of the experimental run.  We refer the interested reader to the attached multimedia file for a movie of the full simulation and experimental runs.\label{fig:experiment}

\begin{figure}
\centering
\subfigure[]{\includegraphics[width=0.45\linewidth]{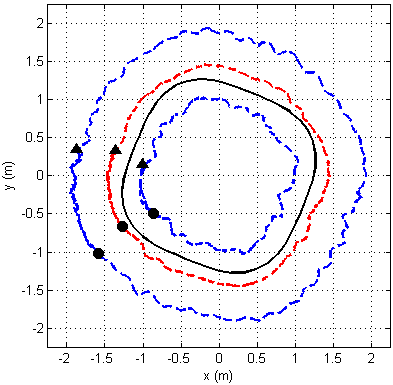}\label{fig:trajectories}}
\subfigure[]{\includegraphics[width=0.45\linewidth]{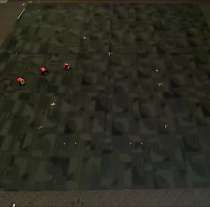}\label{fig:expFrame}}
\caption{Trajectories of the $3$ robot team tracking a (a) star shape.  The red dashed line is the estimated position of the desired boundary shown in solid black.  The start positions are shown by $\triangle$ and the end positions are shown by $\bigcirc$. (b) Snapshot of the multi-robot experiment.\label{fig:expResults}}
\end{figure}

\section{EXTENSION TO PERIODIC BOUNDARIES}\label{sec:doubleGyre}
In this section, we consider the system of $3$ robots with kinematics given by \eqref{eq:kinematics} where $\mathbf{u}_i$ is determined by the wind-driven double-gyre flow model with noise
\begin{subequations}\label{eq:doubleGyre}
\begin{align}
\dot{x} & = -\pi A \sin(\pi \frac{f(x,t)}{s}) \cos(\pi \frac{y}{s}) - \mu x + \eta_1(t), \\
\dot{y} & = \pi A \cos(\pi \frac{f(x,t)}{s}) \sin(\pi \frac{y}{s}) \frac{df}{dx} - \mu y + \eta_2(t), \\
f(x,t) & = \varepsilon \sin(\omega t + \psi) x^2 + (1 - 2 \varepsilon \sin(\omega t+ \psi)) x.
\end{align}
\end{subequations}

When $\varepsilon = 0$, the double-gyre flow is time-independent, while for
$\varepsilon \neq 0$, the gyres undergo a periodic expansion and contraction
in the $x$ direction.  In \eqref{eq:doubleGyre}, $A$ approximately determines
the amplitude of the velocity vectors, $\omega/2\pi$ gives the oscillation
frequency, $\varepsilon$ determines the amplitude of the left-right motion of
the separatrix between the gyres, $\psi$ is the phase, $\mu$ determines the
dissipation, $s$ scales the dimensions of the workspace, and $\eta_i(t)$
describes a stochastic white noise with mean zero and standard deviation
$\sigma = \sqrt{2I}$, for noise intensity $I$.  In this work, $\eta_i(t)$ can be viewed as either measurement or environmental noise.  Fig. \ref{fig:ppDbleGyre} shows the phase portrait of the time-independent double-gyre model.

\begin{figure}
\centering
\includegraphics[width=0.55\linewidth]{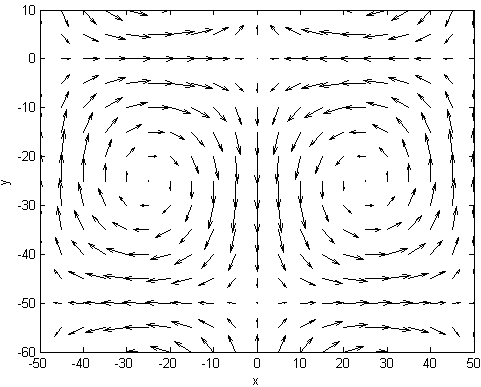}
\caption{Phase portrait of the model given by \eqref{eq:doubleGyre} with $A= 10$, $\mu=0.005$, $\varepsilon = 0$, $\psi = 0$, $I = 0$, and $s = 50$.\label{fig:ppDbleGyre}}
\end{figure}

Fig. \ref{fig:FTLEtrack} shows the use of the control strategy \eqref{eq:straddleControl} and \eqref{eq:trackControl} to track the Lagrangian
coherent structures of the periodic double-gyre model with noise.  As mentioned in Section \ref{sec:intro}, LCS are extensions of stable and unstable manifolds to non-autonomous dynamical systems \cite{ref:Kent2008}.  We note that while the control strategy was based on techniques developed for time-independent systems, the method performs surprisingly well in tracking LCS for slow time-varying systems in the presence of noise.  Details regarding LCS computation can be found in \cite{fbys11} and we refer the interested reader to the attached multimedia file for a movie of the full simulation run.  While the control strategy was developed for static flows, the movie shows the robustness of the strategy for tracking LCS in time-varying flows.

\begin{figure*}[ht]
\centering
\subfigure[t=0.8]{\includegraphics[width =0.24\textwidth]{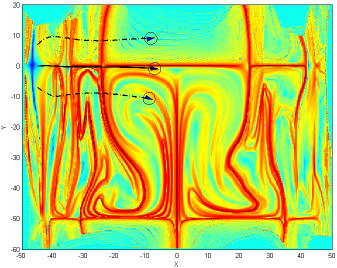}}
\subfigure[t=1.4]{\includegraphics[width =0.24\textwidth]{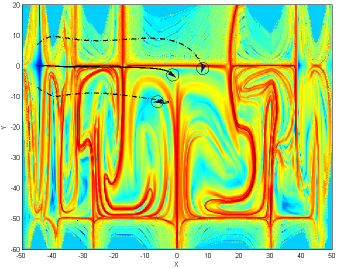}}
\subfigure[t=1.8]{\includegraphics[width =0.24\textwidth]{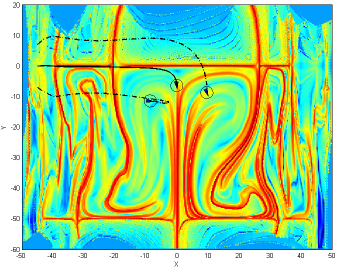}}
\subfigure[t=2.6]{\includegraphics[width =0.24\textwidth]{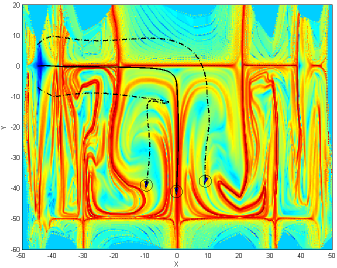}}\\
\subfigure[t=3.0]{\includegraphics[width =0.24\textwidth]{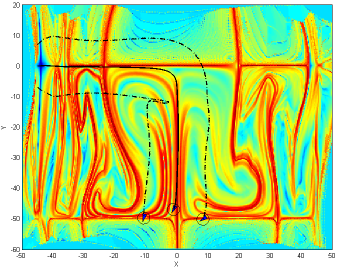}}
\subfigure[t=3.2]{\includegraphics[width =0.24\textwidth]{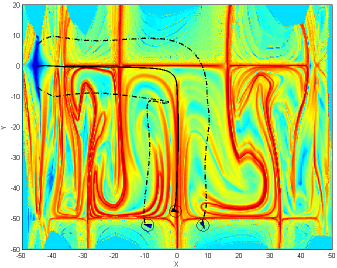}}
\subfigure[t=3.6]{\includegraphics[width =0.24\textwidth]{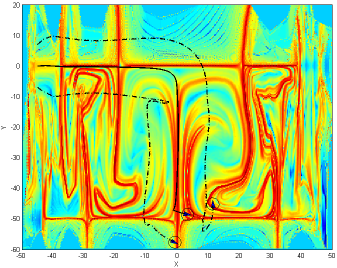}}
\subfigure[t=4.0]{\includegraphics[width =0.24\textwidth]{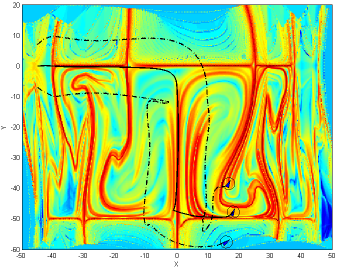}}
\caption{Trajectories of the team of $3$ robots tracking the Lagrangian coherent structures of the system described by \eqref{eq:doubleGyre} with $A= 10$, $\mu=0.005$, $\varepsilon = 0.1$, $\psi = 0$, $I = 0.01$, and $s = 50$.  The trajectories of the straddling robots are shown in black and the estimated LCS is shown in white.\label{fig:FTLEtrack}}
\end{figure*}

\section{DISCUSSION}\label{sec:discuss}
In this paper, we have designed a control strategy that allows collaborating robots to track coherent structures and manifolds on general static conservative flows.  In addition, we showed how the strategy can be used to track LCS in time-dependent conservative flows with measurement noise.  The saddle straddle control strategy is based on the communication of local velocity field measurements obtained by each robot.  Using the local velocity field information provided by the two straddling robots (the herders), one robot (the tracker), is able to detect the coherent structures, a global structure that delineates the phase space into different dynamical regions. Our work is novel in that the robots are determining the location of a global structure based solely on local information, and as far as we know, the sensing of LCS in the ocean has never been performed using autonomous vehicles.  Moreover, only initial state knowledge of the LCS is required locally to get an accurate prediction of the global structure.

While the cooperative control strategy was inspired by the PIM Triple Procedure, a procedure that relies on the computation of escape times which is a global property of the system, the controller itself only relies on information provided by each robot's onboard sensors.  We also note that the realization of the control strategy by the team of robots can be achieved without the need for global localization information. As such, the strategy is a purely local strategy.  Furthermore, the cooperative control strategy was derived to track the manifolds on a static flow, but performs surprisingly well at tracking the LCS in the time-dependent double-gyre model in the presence of noise.

Since realistic quasi-geostrophic ocean models exhibit double-gyre flow solutions, our first attempt seems to suggest that our methods may be and general enough to be applied to more complicated models, including multi-layer PDE ocean models.  As mentioned in Section \ref{sec:analysis} the robustness of the control strategy is dependent on numerous parameters in the system which includes robots' sensing and communication ranges, the bounds on the sensor and actuation noise, the vector interpolation technique, the sampling frequency, and the relative time scales of the AUV dynamics in relation to the surrounding flow dynamics.  While our initial results suggest that our approach may be robust enough to measurement noise, a more thorough understanding of the sensitivity of the proposed strategy to these various parameters is instrumental in extending our approach to more realistic ocean models and for field deployment.

\section{FUTURE WORK}\label{sec:future}
In recent years, there has been significant interest in the use of AUVs to collect scientific data in the ocean to improve our ability to forecast harmful algae blooms and weather and climate patterns.  One drawback to operating sensors in time-dependent and stochastic environments like the ocean is that the sensors will tend to escape from their monitoring region of interest. As such, the ability to identify and track Lagrangian coherent structures (LCS) in these dynamic environments is paramount in maintaining appropriate sensor coverage in regions of interest.  Additionally, since LCS have been shown to coincide with optimal trajectories in the ocean which minimize the energy and navigation time \cite{ref:Inanc2005, ref:Senatore2008}, real-time knowledge of these ``super-highways'' is key in planning efficient AUVs paths.

Of particular interest is the extension of our method to more realistic ocean models.  Specifically, can we extend our current cooperative tracking strategy to a swarm of heterogenous mobile and stationary sensors?  By increasing the team size and incorporating both stationary and mobile sensing devices, it is possible to refine our tracking technique to reveal the coherent structures at various spatial and time scales. One immediate direction for future work is to investigate how the proposed strategy scales to larger team sizes.  Second, underwater environments pose unique challenges in terms of wireless communications.  In general, acoustic transmissions generally have low data rates and acoustic wave propagation can be further affected by the surrounding fluid dynamics \cite{ref:Wang2009}.  As such, a second direction for future work is to investigate how communication delays and missed transmissions impact the overall accuracy of the tracking methodology.  In this work, we assume an initial state knowledge of the LCS is required. This initial formation may be difficult to achieve without any prior global knowledge of the flows.  By considering a team of both stationary and mobile sensors, one can potentially obtain an initial estimate of a local LCS through the stationary sensing network which can then be tracked and further refined by the mobile nodes.  A third direction for future work is to determine how one can strategically place a combination of mobile and stationary sensors to provide real-time updates on the locations of LCS.  





\end{document}